\newif\ifConference
\newif\ifJournal
\newif\ifAnonymous
\newif\ifAppendix
\newif\ifFinal

\Conferencetrue
\Appendixtrue
\Finaltrue

\documentclass[a4paper,USenglish,cleveref,autoref,thm-restate,numberwithinsect,nameinlink]{lipics-v2021}

\usepackage[utf8]{inputenc}
\usepackage{amsmath}
\usepackage{amssymb}
\usepackage{amsthm}
\usepackage{enumerate}
\usepackage{dsfont}
\usepackage{thmtools}
\usepackage{mathtools}
\usepackage{microtype}
\usepackage{graphicx}
\usepackage{multirow}
\usepackage[table]{xcolor}
\usepackage{tikz}
\usetikzlibrary {arrows.meta} 
\usetikzlibrary{graphs}
\usepackage{placeins}
\usepackage{nicefrac}
\usepackage{tabularx}
\usepackage{algorithm2e}
\usepackage[T1]{fontenc}
\usepackage{etoolbox}
\usepackage{hyperref}
\usepackage{booktabs}
\hypersetup{
	colorlinks,
	linkcolor={red!50!black},
	citecolor={blue!90!black!80},
	urlcolor={blue!80!black}
}
\usepackage{cleveref}

\usepackage{lineno}
\ifFinal
\usepackage[disable]{todonotes}
\linenumbers
\else
\usepackage[notref, notcite]{showkeys}
\usepackage{todonotes}
\fi

\usepackage{etoolbox}
\newcommand{\appendixproofs}{}
\newcommand{\toappendix}[1]{\gappto{\appendixproofs}{#1}}
\toappendix{
	\renewcommand\thesection{A}
	\section{Appendix}
	\renewcommand\thesection{A.\arabic{section}}
	\setcounter{section}{0}
}

{\upshape\itshape}{\upshape\rmfamily}

{\upshape\itshape}{\upshape\rmfamily}
{\upshape\itshape}{\upshape\rmfamily}
\usepackage{multirow}

\newcommand{\kommentar}[1]{}
\newcommand{\Oh}{\ensuremath{\mathcal{O}}}

\DeclareMathOperator{\poly}{poly}
\DeclareMathOperator{\tw}{{tw}}
\DeclareMathOperator{\sw}{{sw}}
\DeclareMathOperator{\nsw}{{nsw}}

\newcommand{\yes}{{\normalfont\texttt{yes}}\xspace}

\newcommand{\Wh}[1]{{\normalfont W[\ensuremath{#1}]}\xspace}

\newcommand{\FPT}{{\normalfont{FPT}}\xspace}

\newcommand{\XP}{{\normalfont{XP}}\xspace}

\newcommand{\Instance}{{\ensuremath{\mathcal{I}}}\xspace}

\newcommand{\Food}{{\ensuremath{\mathcal{F}}}\xspace}

\newcommand{\problemdef}[3]{
	\begin{trivlist}
		\item \tikz{\node [draw,inner sep=4.5pt] {\begin{minipage}{\textwidth-9.4pt}
					\normalsize\textsc{#1}
					
					\smallskip
					
          \setlength{\tabcolsep}{0pt}%
					\begin{tabularx}{\textwidth-9.4pt}{l@{\ }>{\raggedright}X}
						\normalsize\textbf{Input:}	  & \normalsize#2 \cr
            \normalsize\textbf{Question:} & \normalsize#3
					\end{tabularx}
			\end{minipage}};}
	\end{trivlist}
}

\newcommand{\PROB}[1]{{{\upshape\textsc{#1}}}\xspace}

\newcommand{\CDSLong}{\PROB{Capacitated Dominating Set}}
\newcommand{\CDS}{\PROB{CDS}}

\ifAnonymous

	\author{Anonymous Author(s)}
	{Anonymous Afiliation(s)}
	{}
	{}
	{}
	
	\authorrunning{Anonymous Author(s)}

\else
	\author{Jannik Schestag}
	{Delft University of Technology, Delft, The Netherlands}
	{j.t.schestag@tudelft.nl}
	{https://orcid.org/0000-0001-7767-2970}
	{Funded by the Dutch Research Council (NWO),\\grant OCENW.GROOT.2019.015,
		``Optimization for and with Machine Learning (OPTIMAL).''}
	\author{Norbert Zeh}
	{Dalhousie University, Halifax, Canada}
	{nzeh@cs.dal.ca}
	{https://orcid.org/0000-0002-0562-1629}
	{Research supported by NSERC Discovery Grant number RGPIN-2025-06235.}
	
	\Copyright{Jannik Schestag and Norbert Zeh}
	
	\authorrunning{J. Schestag and N. Zeh}

\fi



\hideLIPIcs

\usepackage{cite}
\usepackage{hyperref}
\usepackage{doi}
\usetikzlibrary{shadows}

\newcommand{\fsPDD}{\PROB{\mbox{$1$-PDD$_{\text{s}}$}}}

\newcommand{\hsPDD}{\PROB{\mbox{$\nicefrac{1}{2}$-PDD$_{\text{s}}$}}}
\newcommand{\asPDD}{\PROB{\mbox{$\alpha$-PDD$_{\text{s}}$}}}

\newcommand{\wPDDlong}{\PROB{Weighted Phylogenetic Diversity with Dependencies}}
\newcommand{\wPDD}{\PROB{Weighted PDD$_{\text{s}}$}}

\newcommand{\gviable}{\mbox{$\gamma$-viable}\xspace}
\newcommand{\hviable}{\mbox{$\nicefrac12$-viable}\xspace}

\newcommand\B{\mathcal{B}}

\renewcommand\L{\mathcal{L}}
\renewcommand\P{\mathcal{P}}
\newcommand\T{\mathcal{T}}

\newcommand\NN{\mathbb{N}}
\newcommand\RR{\mathbb{R}}

\usepackage{tikz}
\usetikzlibrary{arrows.meta,backgrounds,calc,decorations.pathreplacing,fit,intersections,quotes}
\tikzset{
	>=Latex,
	edge/.style = {draw=black,thick,line cap=butt,line join=round},
	placeholder/.style = {edge,densely dashed},
	chosen/.style = {circle,draw=black,fill=black,inner sep=0pt,minimum size=3pt,outer sep=1pt},
	dot/.style = {circle,draw=black,fill=black,inner sep=0pt,minimum size=1pt},
	not-chosen/.style = {chosen,fill=white},
	fan/.style = {draw=black,thin},
	widget/.style = {thin,draw=black,fill=black!20,rounded corners=3pt}
}

\tikzset{
	fan/.pic = {
		\path [fan]
		(-30:0.2) -- (0,0) -- (30:0.2)
		(-15:0.2) -- (0,0) -- (15:0.2)
		(0,0) -- (0:0.2) coordinate (-label);
}}

\usepackage{hyphenat}
\title       {The First Known Problem That Is FPT with Respect to Node Scanwidth but Not Treewidth}
\titlerunning{Separating Treewidth from Node Scanwidth}

\nolinenumbers

\begin{document}

\maketitle

\begin{abstract}\noindent
  Structural parameters of graphs, such as treewidth, play a central role in
  the study of the parameterized complexity of graph problems.  Motivated by
  the study of parametrized algorithms on phylogenetic networks, scanwidth was
  introduced recently as a new treewidth-like structural parameter for acyclic
  directed graphs (DAGs) that respects the edge directions in the DAG.  The
  utility of this width measure has been demonstrated by results that show that
  a number of problems that are fixed-parameter tractable (FPT) with respect to
  both treewidth and scanwidth allow algorithms with a better dependence on
  scanwidth than on treewidth.  More importantly, these scanwidth-based
  algorithms are often much simpler than their treewidth-based counterparts:
  the name ``scanwidth'' reflects that traversing a tree extension (the
  scanwidth-equivalent of a tree decomposition) of a DAG amounts to
  ``scanning'' the DAG according to a well-chosen topological ordering.

  While these results show that scanwidth is useful especially for solving
  problems on phylogenetic networks, all problems studied through the lens of
  scanwidth so far are either FPT with respect to both scanwidth and treewidth,
  or \Wh{\ell}-hard, for some $\ell \ge 1$, with respect to both.  In this
  paper, we show that scanwidth is not just a proxy for treewidth and provides
  information about the structure of the input graph not provided by treewidth,
  by proving a fairly stark complexity-theoretic separation between these two
  width measures.  Specifically, we prove that \wPDDlong is FPT with respect to
  the scanwidth of the food web but \Wh{\ell}-hard with respect to its
  treewidth, for all $\ell \ge 1$.  To the best of our knowledge, no such
  separation between these two width measures has been shown for any problem
  before.
\end{abstract}

\section{Introduction}

The study of parameterized complexity has shown that, for many computationally
hard problems, the input size alone is a poor predictor for how long it should
take to solve such a problem on a particular input.  It is the \emph{structure}
of the input, expressed by some parameter~$k$ independent of the input size,
that matters more than its size.  A problem $\Pi$ is said to be
\emph{fixed-parameter tractable} (FPT) with respect to this parameter $k$ if
there exists an algorithm that solves this problem on any input of size $n$ and
with parameter $k$ in time $f(k) \cdot \poly(n)$, where $f$ is any computable
function and $\poly(n)$ is any polynomial in $n$ independent of
$k$~\cite{downeybook}. The W-hierarchy provides a classification of problems
believed not to be FPT: Any $\Wh{\ell}$-hard problem, for any $\ell \ge 1$, is
not FPT unless standard complexity-theoretic assumptions
fail~\cite{downeybook}.

When studying graph problems, width measures of graphs have played an important
role in understanding the structure of an input graph. Treewidth is one of the
earliest and most influential width measures~\cite{robertson1986graph} and has
become a cornerstone of algorithm design and structural graph
theory~\cite{cygan,diestel2025graph}.  It is based on a structural
representation of a graph $G$ called a \emph{tree decomposition} of $G$, which
consists of a tree whose nodes have associated small subsets of vertices of $G$
that must satisfy certain conditions. Over the past decades, hundreds of
problems have been shown to be \FPT when parameterized by treewidth,
e.g.~\cite{bodlaender1988dynamic,bodlaender2007treewidth}, while others remain
\Wh1-hard with respect to treewidth~\cite{fellows2011complexity}.  Most of the
FPT-algorithms based on treewidth use dynamic programming over the tree
decomposition.

Phylogenetic networks are directed acyclic graphs (DAGs) representing the
evolution of a set of species. A number of problems on phylogenetic networks
have also been shown to be FPT with respect to treewidth, most notably tree
containment~\cite{vanIerselTreeContainment,bruchhold2025exploiting}, small parsimony~\cite{scornavacca2022treewidth}, and phylogenetic diversity
\cite{PaNDA,MAPPD,MAPPDD}.  Some of these algorithms are rather complex, as a
traversal of the tree decomposition does not correspond to a traversal of the
network that respects edge directions. This motivated the introduction of
scanwidth \cite{berry} as another measure of how tree-like a network is but in
a manner that traversing the tree top-down or bottom-up corresponds to a
traversal of the network in the same direction.  This has led to much simpler
algorithms for a number of problems
\cite{JonesTreeContainment,bruchhold2025exploiting,PaNDA,MAPPDD,scornavacca2022treewidth}.

The scanwidth of a DAG is an upper bound on its
treewidth~\cite{bruchholdThesis}, and so far, simplicity of the algorithms and
a better dependence of the algorithms' running times on scanwidth compared to
parameterizations by treewidth have been the only benefit of scanwidth.  To the
best of our knowledge, no problem was known so far for which there exists a
strong complexity-theoretic separation between these two width measures in the
sense that the problem is FPT with respect to scanwidth and \Wh{1}-hard or
harder with respect to treewidth.


In this paper, we change this by proving that \wPDDlong (\wPDD) is FPT with
respect to scanwidth but even a restricted version of this problem is
\Wh{\ell}-hard with respect to treewidth, for all $\ell \ge 1$. In this
problem, we are given a vertex- and arc-weighted DAG, called a \emph{food web}
in its biological interpretation, together with two integers $B$ and $D$.
Vertices represent species, vertex weights model their individual diversity
scores, and directed arcs represent predator-prey relationships between species
with their weights indicating the importance of the prey for the predator's
survival.  The task is to determine whether there is a set of at most $B$
species whose total diversity is at least $D$ and such that each selected
species has ``sufficient'' prey among the selected species.  We give a formal
definition in \cref{sec:preliminaries}.

Problems of this type arise naturally when combining phylogenetic diversity
with ecological constraints~\cite{1PDDkernel,PDD,moulton,WeightedFW}.
Phylogenetic diversity was originally defined on trees as the total weight of
the minimal subtree spanning a given set of species~\cite{FAITH1992}. While a
subset of $B$ species maximizing phylogenetic diversity can be computed in
polynomial time~\cite{Pardi2005,steel}, incorporating ecological dependencies
makes finding an optimal solution substantially harder~\cite{faller}. To break
this intractability, a simplified model of diversity was introduced in which
each species is assigned an individual diversity score~\cite{faller}, leading
to the formulation of \wPDD and its variants in~\cite{WeightedFW}.

Our first result shows that even a restricted version of \wPDD is intractable
with respect to treewidth. Specifically, we consider \hsPDD, in which each
arc~$(u,v)$ has a weight of $2/\deg^-(v)$, and prove the following theorem:

\begin{theorem}
	\label{thm:tw-hard}
  \hsPDD is $\Wh\ell$-hard, for all $\ell \ge 1$, when parameterized by the
  treewidth of the food web.
\end{theorem}

This result complements a previously known \XP algorithm for \hsPDD
parameterized by treewidth~\cite{WeightedFW} and rules out the existence of an
\FPT algorithm  for this problem under standard complexity-theoretic
assumptions. In contrast, we show that \wPDD is FPT with respect to node
scanwidth:

\begin{theorem}
	\label{thm:nsw-fpt}
  \wPDD can be solved in~$\Oh\left(2^{\nsw(\Food)} \cdot n^3\right)$ time if a
  tree extension of the food web $\Food$ of node scanwidth $\nsw(\Food)$ is
  given. Here, $n$ is the number of species, that is, the number of vertices in
  $\Food$.
\end{theorem}


\section{Preliminaries}

\label{sec:preliminaries}

\subparagraph*{Graphs and functions.}

We use standard graph terminology as in~\cite{diestel2025graph}. We use
$\deg^-_G(v)$ to denote the in-degree of a vertex $v$ in a graph $G$ and omit
the subscript when the graph is clear from context. In addition, we use the
terms \emph{vertex} and \emph{species} interchangeably, given that vertices in
food webs represent species.

Given a function $f : A \to \RR$ and a subset $B \subseteq A$, we define $f(B)
:= \sum_{b \in B} f(b)$.

\subparagraph*{Food webs, viability, and phylogenetic diversity.}

A \emph{food web~$\Food=(X,A)$} on a set of species~$X$ is a directed acyclic
graph (DAG). A species $x$ is a \emph{prey} of a species $y$, and $y$ is a
\emph{predator} of $x$, if~$(x,y) \in A$. Sources are species without prey.

Given a food web~$\Food=(X,A)$, we define an arc weight function~$\gamma: A \to
(0,1]$, where~$\gamma(u,v)$ measures the importance of $u$ for the survival of
$v$. A set of species~$S \subseteq X$ is~\emph{$\gamma$-viable}
if~$\gamma_{v}(S) := \gamma\{ (u,v) \in A \mid u\in S \} \ge 1$, for every
non-source vertex $v \in S$. We say~$S$ is \emph{viable} if~$\gamma$ is clear
from context.  Intuitively, this means that if we make efforts to save the
species in~$S$ from extinction, then for every species we save, we also save
enough of its prey to enable it to survive.

Phylogenetic diversity is a measure of the diversity of the genetic material of
a set of species.  This is usually represented by assigning weights to the
edges of the phylogenetic tree representing the evolution of this set of
species.  In the version of the phylogenetic diversity problem we study, the
phylogenetic tree is a star, which means that we can represent the contribution
of each species to the total diversity of a subset of $X$ as a vertex weight
function~$d : X \to \NN^+$.

\problemdef{\textsc{Weighted Phylogenetic Diversity with Dependencies} (\wPDD)}{
  A food-web~$\Food = (X,A)$, a diversity function $d : X \to \NN^+$, an
  arc weight function~$\gamma: A \to (0,1]$, and integers~$B$ and~$D$.
}{
	Is there a \gviable set~$S\subseteq X$ such that~$|S|\le B$ and $d(S)\ge D$?
}

\asPDD, for~$\alpha \in (0,1]$, is a special case of \wPDD in which there is no
arc weight function and a set $S$ is viable if $S$ contains at least $\alpha
\cdot \deg^-(v)$ of the prey of every species $v \in S$.  This is equivalent to
setting $\gamma(u,v) = 1/(\alpha \cdot \deg^-(v))$ for each arc~$(u,v) \in
\Food$.

\subparagraph*{Treewidth and scanwidth.}

A \emph{tree decomposition} of an undirected graph~$G$ consists of a rooted
tree~$T$ and, for every node~$t \in T$, a bag $Q_t \subseteq V(G)$ such that,
for every edge $\{u,v\} \in E(G)$, there exists a bag $Q_t$ with $u,v \in Q_t$
and, for every vertex $v \in V(G)$, the nodes $t \in T$ such that $v \in Q_t$
induce a subtree of $T$. The \emph{width} of a tree decomposition is the size
of its largest bag minus one.  The \emph{treewidth} $\tw(G)$ of an undirected
graph $G$ is the minimum width of all tree decompositions of $G$
\cite{robertson1986graph}. The treewidth of a directed graph is the treewidth
of its underlying undirected graph.

A \emph{tree extension} of a DAG~$G=(V,A)$ is a rooted tree~$T=(V,A')$ on the
set of vertices of $G$ such that, for every arc~$(u,v)\in A$, $u$ is a (proper)
ancestor of $v$ in~$T$ \cite{berry}.  The \emph{node width} of a vertex $v \in
T$ is the number of its proper ancestors in $T$ that have out-neighbours in~$G$
that are descendants of $v$ in $T$ (including $v$ itself).  The node width of
$T$ is the maximum node width of all its vertices. Similarly, the \emph{edge
width} of a node $v \in T$ is the number of arcs~$(u,w) \in G$ such that $u$ is
a proper ancestor of $v$ and $w$ is a descendant of~$v$ in~$T$, and the edge
width of $T$ is the maximum edge width of all its vertices. The
\emph{scanwidth} $\sw(G)$ and \emph{node scanwidth} $\nsw(G)$ of a DAG $G$ are
the minimum edge width and node width, respectively, of all tree extensions of
$G$~\cite{berry,holtgrefeThesis,holtgrefe2024exact}. These definitions
immediately imply that $\nsw(G) \le \sw(G)$, for every DAG $G$. Moreover,
$\tw(G) \le \nsw(G)$~\cite[Lemma~3.3]{bruchholdThesis}.

\subparagraph*{Parameterized complexity.}

\ifConference
A \emph{parameterization} of a formal language (decision problem) $\L \subseteq
\Sigma^*$ over an alphabet $\Sigma$ is a function $\kappa : \Sigma^* \to \NN$.
$\L$ is \emph{fixed-parameter-tractable} (FPT) with respect to $\kappa$ if
there exists an algorithm that takes $f(\kappa(\sigma)) \cdot p(|\sigma|)$ time
to decide whether~$\sigma \in \L$, where $f(\cdot)$ is any computable function
and $p(\cdot)$ is a polynomial. See \cite{cygan} for a detailed introduction.
\Wh{\ell}-hard problems, for $\ell \ge 1$, are believed not to
be~\FPT~\cite{downeybook}. In this paper, we are interested in
parameterizations of \wPDD by the treewidth or node scanwidth of the food web,
that is $\kappa((\Food,d,\gamma,B,D)) = \tw(\Food)$ or
$\kappa((\Food,d,\gamma,B,D)) = \nsw(\Food)$, for every instance
$(\Food,d,\gamma,B,D)$ of \wPDD.
\fi

\section{\boldmath \hsPDD Parameterized by Treewidth}

In this section, we prove \Cref{thm:tw-hard}:
\hsPDD is \Wh{\ell}-hard with respect to treewidth, for all~$\ell\in\NN_{\ge 1}$.
We note that this proof can easily be adjusted to show hardness of \asPDD,
for any~$\alpha \in (0,1)$.
In contrast, \fsPDD is \FPT when parameterized by treewidth~\cite{WeightedFW}.

Our proof uses a reduction from \CDSLong, which is the following:
%
\problemdef{\CDSLong (\CDS)}{
	An undirected graph~$G = (V, E)$, a capacity function~$c : V \to \NN$, and
	an integer $k \in \NN$.
}{
  Is there a pair~$(S,f)$, where $S \subseteq V$ is a set of at most $k$
  vertices, and $f : V \setminus S \to S$ is a function that satisfies
	$|f^{-1}(v)| \le c(v)$ for all $v \in S$?
}
%
Intuitively, $u \in S$ is ``used to dominate'' $v \in V \setminus S$ if $f(v) =
u$, and the capacity constraints limit the number of vertices each vertex in
$S$ can be used to dominate. \CDS is \Wh\ell-hard with respect to the graph's
pathwidth, for all $\ell \ge 1$~\cite{bodlaender2025xnlp}, and therefore also
with respect to its treewidth.
Given an instance $\Instance := (G,c,k)$ of \CDS, we construct, in polynomial time, an instance
$\Instance' := (\Food,d,B,D)$ of~\hsPDD such that $|\Food| \in \poly(|G|)$,
$\tw(\Food) \in \Oh(\tw(G))$, and \Instance is a \yes-instance of
\CDS if and only if~$\Instance'$ is a \yes-instance of \hsPDD.
This proves \cref{thm:tw-hard}.

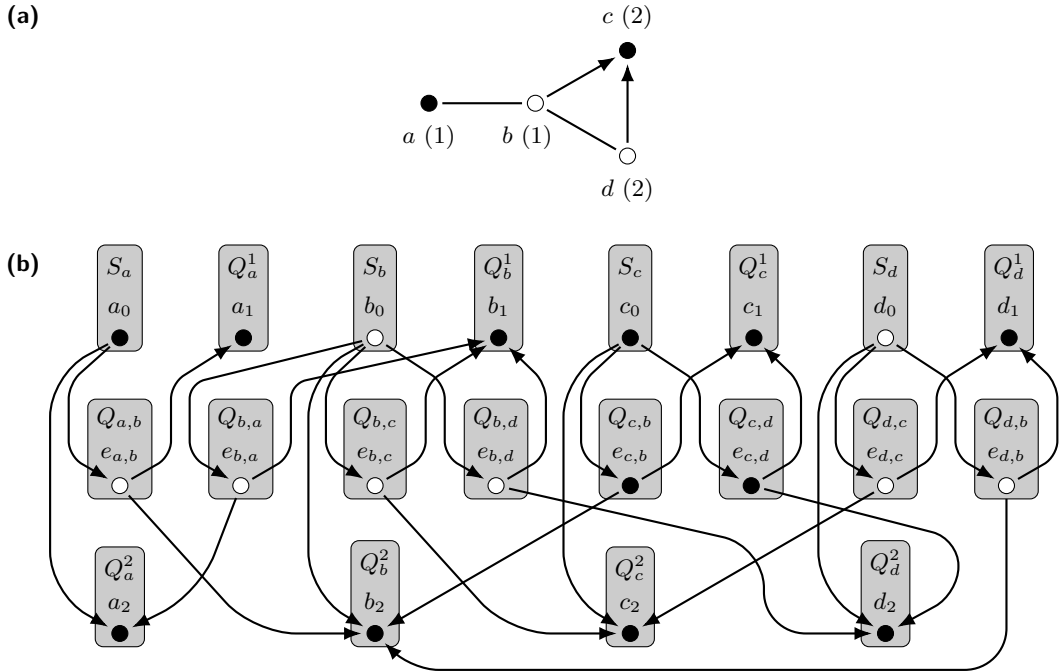
\begin{figure}[t]
	\centering
	\begin{tikzpicture}[x=14mm,y=14mm,remember picture]
      \small
			\path
			node [scale=2,chosen,label=below:$a$ (1)] (a) {}
			++(0:1) node [scale=2,not-chosen,label={[xshift=-3pt]below:$b$ (1)}] (b) {}
			+(30:1) node [scale=2,chosen,label=above:$c$ (2)] (c) {}
			+(330:1) node [scale=2,not-chosen,label=below:$d$ (2)] (d) {};
			\path [edge]
			(a) -- (b) -- (d);
			\path [edge,->]
			(b) edge (c) (d) edge (c);
    \path (current bounding box.north west) coordinate (a-top-left);
	\end{tikzpicture}\\[\bigskipamount]
	\begin{tikzpicture}[x=14mm,y=14mm,remember picture]
      \small
			\begin{scope}[overlay]
				\node (measure) {\phantom{$Q_{b,d}$}};
				\node [fit=(measure)] (measure) {};
				\path let \p0 = (measure.west),
				\p1 = (measure.east),
				\n0 = {\x1 - \x0} in
				foreach \x/\y/\is/\shift in {a/b/not-/0,b/a/not-/0.375,b/c/not-/0.5,b/d/not-/0.375,c/b//0.5,c/d//0.375,d/c/not-/0.5,d/b/not-/0.375} {
					++(0:\shift) ++(0:\n0) node [scale=2,\is chosen,label={[name=le\x\y]above:$e_{\x,\y}$}] (e\x\y) {}
					node [above=0mm of le\x\y] (Q\x\y) {$Q_{\x,\y}$}
				}
				foreach \x/\is/\ref in {a//ab,b/not-/bc,c//cb,d/not-/dc} {
					(e\ref) ++(90:1.4) node [scale=2,\is chosen,label={[name=l\x0]above:\vphantom{$b_0$}$\x_0$}] (\x0) {}
					++(0:0.4) ++(0:\n0) node [scale=2,chosen,label={[name=l\x1]above:\vphantom{$b_1$}$\x_1$}] (\x1) {}
				}
				foreach \x in {a,b,c,d} {
					(\x0) +(270:2.8) node [scale=2,chosen,label={[name=l\x2]above:$\x_2$}] (\x2) {}
					node [above=0mm of l\x0] (S\x) {\vphantom{$Q_d^1$}$S_\x$}
					node [above=0mm of l\x1] (Q\x1) {\vphantom{$Q_d^1$}$Q_\x^1$}
					node [above=0mm of l\x2] (Q\x2) {\vphantom{$Q_d^2$}$Q_\x^2$}
				};
			\end{scope}
			\begin{scope}[on background layer]
				\path foreach \x in {ab,ba,bc,bd,cb,cd,db,dc} {
					node [widget,fit={(e\x) (Q\x)},inner sep=0pt] (Q\x) {}
				}
				foreach \x in {a,b,c,d} {
					node [widget,fit={(S\x) (\x0)},inner sep=0pt] (S\x) {}
					node [widget,fit={(Q\x1) (\x1)},inner sep=0pt] (Q\x1) {}
					node [widget,fit={(Q\x2) (\x2)},inner sep=0pt] (Q\x2) {}
				};
			\end{scope}
			\path
			($(Qab.east)!0.333!(Qba.west)$) coordinate (xab1)
			($(Qab.east)!0.667!(Qba.west)$) coordinate (xba0)
			($(Qba.east)!0.25!(Qbc.west)$) coordinate (xba1)
			($(Qba.east)!0.5!(Qbc.west)$) coordinate (xb0b2)
			($(Qba.east)!0.75!(Qbc.west)$) coordinate (xbc0)
			($(Qbc.east)!0.333!(Qbd.west)$) coordinate (xbc1)
			($(Qbc.east)!0.667!(Qbd.west)$) coordinate (xbd0)
			($(Qbd.east)!0.25!(Qcb.west)$) coordinate (xbd1)
			($(Qbd.east)!0.5!(Qcb.west)$) coordinate (xc0c2)
			($(Qbd.east)!0.75!(Qcb.west)$) coordinate (xcb0)
			($(Qcb.east)!0.333!(Qcd.west)$) coordinate (xcb1)
			($(Qcb.east)!0.667!(Qcd.west)$) coordinate (xcd0)
			($(Qcd.east)!0.25!(Qdc.west)$) coordinate (xcd1)
			($(Qcd.east)!0.5!(Qdc.west)$) coordinate (xd0d2)
			($(Qcd.east)!0.75!(Qdc.west)$) coordinate (xdc0)
			($(Qdc.east)!0.333!(Qdb.west)$) coordinate (xdc1)
			($(Qdc.east)!0.667!(Qdb.west)$) coordinate (xdb0);
			\path let \p0 = (Qab.east),
			\p1 = (xab1),
			\n0 = {\x1 - \x0} in
			(Qab.west) +(180:\n0) coordinate (xab0)
			(Qdb.east) +(0:\n0) coordinate (xdb1)
			(Qab.north) +(90:2mm) coordinate (bend-y)
			($(Qab.north)!0.5!(Sa.south)$) coordinate (bend-y2)
			(Qab.west) +(180:2*\n0) coordinate (xa0a2);
      \begin{scope}[overlay]
        \path [name path=a] (xab0) -- +(270:1);
        \path [name path=b] (eab) -- +(150:1);
        \path [name intersections={of=a and b}] (intersection-1) coordinate (bend-y3);
        \path [name path=a] (xa0a2) -- +(270:2);
        \path [name path=b] (a2) -- +(150:1);
        \path [name intersections={of=a and b}] (intersection-1) coordinate (bend-y4);
      \end{scope}
			\begin{scope}[rounded corners=2mm]
				\foreach \x/\y in {a/b,b/a,b/c,b/d,c/b,c/d,d/b,d/c} {
					\path [edge,->] (\x0) -- (x\x\y0 |- bend-y) -- (x\x\y0 |- bend-y3) -- (e\x\y);
					\path [edge,->] (e\x\y) -- (x\x\y1 |- bend-y3) -- (x\x\y1 |- bend-y) -- (\x1);
				}
			\end{scope}
			\begin{scope}[rounded corners=5.5mm]
				\foreach \x in {a,b,c,d} {
					\path [edge,->] (\x0) -- (x\x0\x2 |- bend-y2) -- (x\x0\x2 |- bend-y4) -- (\x2);
				}
				\path
				(Qa2.north) +(90:1mm) coordinate (bend-y)
				($(Qa2.east)!0.333!(xb0b2)$) coordinate (xba)
				($(Qa2.east)!0.667!(xb0b2)$) coordinate (xab)
				($(Qb2.east)!0.333!(xc0c2)$) coordinate (xcb)
				($(Qb2.east)!0.667!(xc0c2)$) coordinate (xbc)
				($(Qc2.east)!0.333!(xd0d2)$) coordinate (xdc)
				($(Qc2.east)!0.667!(xd0d2)$) coordinate (xbd)
				($(Qd2.east)!0.5!(edb)$) coordinate (xcd);
				\path [edge,->] (eba) -- (xba |- bend-y4) -- (a2);
				\path [edge,->] (ecb) -- (b2);
				\path [edge,->] (edc) -- (c2);
				\path [edge,->,rounded corners=3.5mm] (ecd) -- (xcd |- bend-y) -- (xcd |- bend-y4) -- (d2);
				\path [edge,->]
				(eab) -- (xab |- b2) -- (b2);
				\path [edge,->]
				(ebc) -- (xbc |- c2) -- (c2);
				\path [edge,->]
				(ebd) -- (xbd |- Qa2.north) -- (xbd |- d2) -- (d2);
				\path [edge,->]
				let \p0 = (Qb2.south),
				\p1 = (b2),
				\n0 = {\y0 - 3mm},
				\n1 = {\x1 + (\y1 - \n0)} in
				(\n1,\n0) coordinate (bend-db)
        (edb) -- (edb |- bend-db) {[rounded corners=3.5mm] -- (bend-db) -- (b2)};
			\end{scope}
    \path (current bounding box.north west) node [below left] (b-label) {\sffamily\bfseries (b)};
    \begin{scope}[overlay]
      \path (b-label |- a-top-left) node [below] {\sffamily\bfseries (a)};
    \end{scope}
	\end{tikzpicture}%
  \caption{The reduction from \CDS to \hsPDD.  (a) An input graph~$G$ of \CDS.
    The capacities are shown beside each vertex.  Solid vertices are in a
    dominating set $S$; hollow vertices are not.  Arrows on edges represent
    the mapping $f$ from vertices in $V \setminus S$ to $S$.  (b) The food web
    \Food constructed from $G$. Only the root of each widget is shown.  Solid
    vertices are in a viable set $S'$ corresponding to $(S,f)$.  Hollow
  vertices are not.}
	\label{fig:reduction}
\end{figure}

We build $\Food$ from two types of widgets, as illustrated in
\cref{fig:reduction}: \emph{selector widgets} and \emph{quota widgets}.  We
give their formal definitions after giving an overview of the reduction.  Each of
these widgets is a tree with all its arcs directed towards the root.  The only
vertex in each widget with neighbours outside the widget is the root. For each
vertex $v \in V(G)$, $\Food$ contains one selector widget $S_v$ with root $v_0$
and two quota widgets $Q_v^1$ and $Q_v^2$ with roots~$v_1$ and~$v_2$,
respectively. For every edge $\{u,v\} \in E(G)$, $\Food$ contains two quota
widgets~$Q_{u,v}$ and~$Q_{v,u}$ with roots $e_{u,v}$ and $e_{v,u}$,
respectively.  We add an arc $(v_0,v_2)$ for each vertex~$v \in V(G)$, and
arcs~$(u_0,e_{u,v})$, $(e_{u,v},u_1)$, and $(e_{v,u},u_2)$ for each edge $\{u,v\}
\in E(G)$.  We assign diversity scores to the vertices in $\Food$ and choose a
target diversity $D$ that force any solution of \hsPDD on $\Instance'$ to
contain $v_1$ and $v_2$, for every vertex $v \in G$, as well as at least~$k$ of
the vertices in~$X = \{v_0 \mid v \in G\}$.  The selector widgets ensure that,
when we add a vertex~$v_0$ to a viable set in $\Food$, a certain number of
vertices in $S_v$ also need to be added. Together with a carefully chosen
budget, this ensures that at most $k$ vertices in~$X$ are in a solution, that
is, we choose exactly $k$ vertices in~$X$, representing a dominating set $S$ of
size $k$.  The quota widgets are designed to ensure that adding the root of
such a widget to the viable set requires us to add a certain number of
additional vertices inside the widget; if we add some prey of the root outside
the quota widget, then the number of vertices inside the quota widget that need
to be added decreases. For each vertex $v_2$, we choose the parameters of the
quota widget~$Q_v^2$ so that adding $v_2$ without adding any of its prey
outside $Q_v^2$ is impossible within our budget.  Thus, $v_2$ can be added only
if $v_0$ or some vertex $e_{u,v}$, for an edge $\{u,v\}$ incident to~$v$, is
added. The parameters of the quota widget $Q_{u,v}$ are chosen so that this
requires $u_0$ to be added as well. Thus, each vertex $v$ is either in $S$ or
has a neighbour $f(v) = u$ in~$S$, that is, $S$ is a dominating set. Finally,
adding the vertices $u_1$, for all $u \in G$, without adding any of the
vertices $e_{u,v}$ corresponding to the edges of $G$ would require us to add a
certain fixed number of vertices from the quota widgets $Q_u^1$, for all $u \in
G$. Whenever we add a vertex $e_{u,v}$, the required number of vertices in
$Q_u^1$ decreases, and the parameters of the quota widgets are chosen so that
this decrease in the cost of adding $u_1$ completely offsets the cost of
adding~$e_{u,v}$.  However, the structure of the quota widgets ensures that this
offset of the cost of adding edge vertices~$e_{u,v}$ is realized only for up to
$c(u)$ such ``edge vertices'' corresponding to edges incident to $u$.
Protecting additional edge vertices forces us to exceed the budget.  This
ensures that at most $c(u)$ edge vertices $e_{u,v}$ are added for each vertex
$u \in G$, that is, $u$ is used to dominate at most $c(u)$ vertices in $G$.

It remains to define the selector and quota widgets formally and formalize this
construction. We start by defining the following values, where $n$ and $m$ are
the numbers of vertices and edges in $G$, respectively:
\begin{align*}
  C &:= 3m + 1, && \text{(parameter used to define the other values)}\\
  B &:= C \cdot k + C \cdot \sum_{v \in G}c(v) + 2n + 4m,&& \text{(the budget)}\\
  D_0 &:= 2n \cdot (C + B) + 2m \cdot (2C + 5),&&\text{(diversity of vertices in dominating set)}\\
  D_1 &:= (n + 1) \cdot D_0,&&\text{(diversity of dominated vertices)}
\end{align*}
The values $C$ and~$B$ are also used to define the in-degrees of certain
vertices in $\Food$.

Every selector widget $S_v$ consists of its root $v_0$ and $2(C - 1)$
additional vertices, which are prey of $v_0$.  This is illustrated in
\cref{fig:selector-widget}.

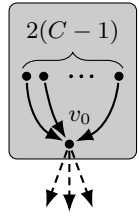
\begin{figure}
	\centering
  \begin{tikzpicture}[x=9mm,y=9mm]
    \small
		\path
		node [chosen,label={[xshift=4pt,yshift=2pt]above:$v_0$}] (r) {}
		++(90:1) ++(180:0.375) node [chosen] (b) {}
		+(180:0.25) node [chosen] (a) {}
		++(0:0.4) node [dot] {}
		++(0:0.15) node [dot] {}
		++(0:0.15) node [dot] {}
		+(0:0.4) node [chosen] (c) {};
		\path [edge,->]
		(a) edge [bend right] (r)
		(b) edge [bend right=10] (r)
		(c) edge [bend left] (r);
		\path [placeholder,->]
		(r) edge +(250:1) edge +(270:1) edge +(290:1);
		\path
		(a) +(-1.5pt,3pt) coordinate (w)
		(c) +(1.5pt,3pt) coordinate (e)
		($(w)!0.5!(e)$) +(90:6pt) node [anchor=south] (C) {$2(C-1)$};
		\draw [decorate,decoration={brace,amplitude=6pt}] (w) -- (e);
		\begin{scope}[on background layer]
			\node [widget,fit=(r) (a) (c) (C)] {};
		\end{scope}
	\end{tikzpicture}
	\caption{A selector widget $S_v$.  Dashed arcs indicate predators of $v_0$
		outside of $S_v$.  There are no other arcs to or from vertices outside of $S_v$.
	}
	\label{fig:selector-widget}
\end{figure}

\begin{observation}
	\label{obs:selector-widget-size}
	A selector widget of type $S(C)$ has $2C - 1$ vertices and $2C - 2$ arcs.
\end{observation}

\begin{observation}
	\label{obs:selector-widget}
  Any viable set $P$ in $\Food$ that contains~$v_0$ contains at least~$C$
  vertices in~$S_v$.  Conversely, for any viable set $P$, there exists a viable
  set $P'$ that contains the same vertices as $P$ outside of $S_v$, as well as
  $v_0$ and $C - 1$ additional vertices in $S_v$.
\end{observation}

Every quota widget $Q$ with root $x$ has a type $Q(\delta,\ell,r,M)$; $\delta$
is the number of prey of $x$ outside of $Q$, $r$ is a parameter that ensures
that any viable set that contains $x$ must also contain at least
$\max(\delta,r)$ prey of $x$, $\ell$ specifies the number of prey among them
that must be outside of $Q$, and the addition of each of a subset of ``medium''
prey of $x$ to a viable set requires the addition of $M$ vertices in $Q$ to
such a set. There are also ``big'' prey of $x$ in $Q$ that cannot be added to a
viable set of size at most $B$ that contains $x$. Their purpose is to increase
the number of prey of $x$ to force a certain number of prey of $x$ to be
contained in any viable set.  More precisely, a quota widget of type
$Q(\delta,\ell,r,M)$ consists of $x$, $\ell + \max(0,r - \delta)$ ``big''
vertices, $r - \ell$ ``medium'' vertices, and $\max(0,\delta - r)$ ``small''
vertices, all of which are prey of $x$. Small vertices have no prey of their
own. Medium vertices have $2(M - 1)$ prey, so adding such a vertex to a viable
set requires us to add at least $M - 1$ of its prey. Big vertices have $2(B-1)$
prey, so any viable set that contains $x$ and at least one such vertex would
exceed the budget $B$.  This definition is illustrated in \cref{fig:quota-widget}.

\begin{observation}
	\label{obs:quota-widget-size}
	A quota widget~$Q$ of type~$Q(\delta,\ell,r,M)$ contains
	\begin{equation*}
		|V(Q)| := (2B - 1)\cdot (\ell + \max(0,r-\delta)) + (2M-1) \cdot (r-\ell) + \max(0,\delta-r) + 1
	\end{equation*}
	vertices and $|V(Q)| -1$ arcs.
  Including the $\delta$ prey outside of $Q$, the root of~$Q$ has~$2 \cdot \max(\delta,r)$ prey.
\end{observation}

\begin{lemma}
	\label{lem:quota-widget}
  Let $Q$ be a quota widget in $\Food$ of type~$Q(\delta,\ell,r,M)$, and let $x$ be
  its root.  Every viable set $P$ in $\Food$ of size $|P| \le B$ that contains
  $x$ contains $\ell' \ge \ell$ of the prey of~$x$ outside of $Q$ and at least
  $1 + \max(\delta,r) - \ell' + (M - 1) \cdot \max(0,r - \ell')$ vertices in $Q$.
  Conversely, for any viable set $P$ that contains $\ell' \ge \ell$ of the
  prey of $x$ outside of $Q$, there exists a viable set $P'$ that contains
  the same vertices as $P$ outside of $Q$, as well as $x$ and exactly $\max(\delta,r)
  - \ell' + (M - 1) \cdot \max(0, r - \ell')$ additional vertices in $Q$.
\end{lemma}

\begin{proof}
	We prove the second claim first. Let~$P$ be a
	viable set that contains~$x$ and~$\ell'\ge\ell$ prey of $x$ outside of
	$Q$.  As $x$ has $2\cdot \max(\delta,r)$ prey overall, $P$ contains at
	least~$t:=\max(\delta,r)$ of its prey.  To construct $P'$, we remove all vertices
	in $V(Q)$ from $P$.  Then, we add $x$, $\min(s, t - \ell')$ of the $s =
	\max(0, \delta - r)$ small vertices and $\max(0,r - \ell')$ medium vertices
	to~$P'$, together with~$M-1$ of the prey of each added medium vertex.  Since
	every vertex outside of $Q$ has the same prey in $P'$ as in $P$ and it is easily
	checked that every vertex in $P' \cap V(Q)$ has at least half its prey in
	$P'$, $P'$ is viable.  The total number of vertices in $P' \cap V(Q)$,
	including $x$, is $1 + \min(s,t - \ell') + M \cdot \max(0,r - \ell') = 1 +
	\max(\delta,r) - \ell' + (M - 1) \cdot \max(0,r - \ell')$.
	
	To prove the first claim, note that no viable set $P$ of size at most $B$
	that contains $x$ contains any big vertices because including such a vertex
	would require at least $B - 1$ of its prey to be included as well, that is,
	$P$ would contain at least $B + 1$ vertices. Since there are only
	$\max(\delta,r) - \ell$ small and medium vertices in~$Q$, and $P$ must contain
	$t = \max(\delta,r)$ prey of $x$, it must contain $\ell' \ge \ell$ prey of
	$x$ outside of $Q$ and $\max(\delta,r) - \ell'$ small and medium vertices in
	$Q$. As there are only $s = \max(0,\delta-r)$ small vertices, $P$ contains at
	least $\max(0,t - \ell' - s) = \max(0, r - \ell')$ medium vertices.  For each
	medium vertex, at least $M - 1$ of its prey have to be in~$P$. That is, $P$
	contains at least $1 + t - \ell' + (M - 1) \cdot \max(0, r - \ell') = 1 +
	\max(\delta,r) - \ell' + (M - 1) \cdot \max(0, r - \ell')$ vertices in~$Q$.
\end{proof}

\begin{figure}
	\centering
	\begin{tikzpicture}[x=9mm,y=9mm]
		\small
		\path
		foreach \t in {b,m} {
			foreach \i/\s in {1/0.5,2/1.5,3/0.5} {
				node [chosen] (\t\i1) {}
				++(0:0.25) node [chosen] (\t\i2) {}
				++(0:0.4) node [dot] {}
				++(0:0.15) node [dot] {}
				++(0:0.15) node [dot] {}
				++(0:0.4) node [chosen] (\t\i3) {}
				++(0:\s)
			}
		}
		foreach \t in {b,m} {
			foreach \i in {1,2,3} {
				($(\t\i1)!0.5!(\t\i3)$) +(270:1) node [chosen] (\t\i) {}
			}
		}
		(m3 -| m33) ++(0:0.75) node [chosen] (s1) {}
		++(0:0.25) node [chosen] (s2) {}
		++(0:0.4) node [dot] {}
		++(0:0.15) node [dot] {}
		++(0:0.15) node [dot] {}
		++(0:0.4) node [chosen] (s3) {}
		($(b11)!0.5!(s3)$) +(270:2) node [chosen] (r) {}
		($(b2)!0.5!(b3)$) node [dot] {}
		+(180:0.15) node [dot] {}
		+(0:0.15) node [dot] {}
		($(m2)!0.5!(m3)$) node [dot] {}
		+(180:0.15) node [dot] {}
		+(0:0.15) node [dot] {}
		(r) +(-13pt,21pt) node (lr) {$x$};
		\draw [pin edge] ([xshift=3pt]lr.south) -- (r);
		\path [edge,->]
		foreach \t in {b,m} {
			foreach \i in {1,2,3} {
				(\t\i1) edge [bend right] (\t\i)
				(\t\i2) edge [bend right=10] (\t\i)
				(\t\i3) edge [bend left] (\t\i)
			}
		}
		(b1) edge [in=180,out=270,out distance=10mm] (r)
		(b2) edge [in=157.5,out=300,out distance=10mm] (r)
		(b3) edge [in=135,out=330] (r)
		(m1) edge [in=112.5,out=250] (r)
		(m2) edge [in=90,out=190,in distance=10mm] (r)
		(m3) edge [in=67.5,out=200,in distance=12mm] (r)
		(s1) edge [in=45,out=240,in distance=15mm,out distance=9mm] (r)
		(s2) edge [in=22.5,out=250,out distance=10mm] (r)
		(s3) edge [in=0,out=270,out distance=10mm] (r);
		\path [placeholder,<-]
		(r) edge +(216:1) edge +(234:1) edge +(252:1);
		\path [placeholder,->]
		(r) edge +(288:1) edge +(306:1) edge +(324:1);
		\foreach \t/\count in {b/2(B-1),m/2(M-1)} {
			\foreach \i in {1,2,3} {
				\path
				(\t\i1) +(-1.5pt,3pt) coordinate (w)
				(\t\i3) +(1.5pt,3pt) coordinate (e)
				($(w)!0.5!(e)$) +(90:6pt) node [anchor=south] (\t-count-\i) {$\count$};
				\draw [decorate,decoration={brace,amplitude=6pt}] (w) -- (e);
			}
		}
		\path
		(s1) +(-1.5pt,3pt) coordinate (w)
		(s3) +(1.5pt,3pt) coordinate (e)
		($(w)!0.5!(e)$) +(90:6pt) node [anchor=south] (s-count) {\shortstack{\underline{Small}\\[3pt]$\max(0,\delta - r)$}};
		\draw [decorate,decoration={brace,amplitude=6pt}] (w) -- (e);
		\path
		(b-count-1.north -| b11) +(-1.5pt,0pt) coordinate (w)
		(b-count-3.north -| b33) +(1.5pt,0pt) coordinate (e)
		($(w)!0.5!(e)$) +(90:6pt) node [anchor=south] (b-count) {\shortstack{\underline{Big}\\[3pt]$\ell + \max(0,r - \delta)$}};
		\draw [decorate,decoration={brace,amplitude=6pt}] (w) -- (e);
		\path
		(m-count-1.north -| m11) +(-1.5pt,0pt) coordinate (w)
		(m-count-3.north -| m33) +(1.5pt,0pt) coordinate (e)
		($(w)!0.5!(e)$) +(90:6pt) node [anchor=south] (b-count) {\shortstack{\underline{Medium}\\[3pt]$r - \ell$}};
		\draw [decorate,decoration={brace,amplitude=6pt}] (w) -- (e);
		\begin{scope}[on background layer]
			\path (b-count.north) +(90:3pt) coordinate (top);
			\node [widget,fit={(top) (b11) (s-count) (r)}] {};
		\end{scope}
	\end{tikzpicture}
	\caption{A quota widget $Q$ of type~$Q(\delta,\ell,r,M)$ and with root $x$.
		Dashed arcs indicate connections to neighbours of $x$ outside of $Q$.  There
		are no other arcs to or from vertices outside of $Q$.
	}
	\label{fig:quota-widget}
\end{figure}
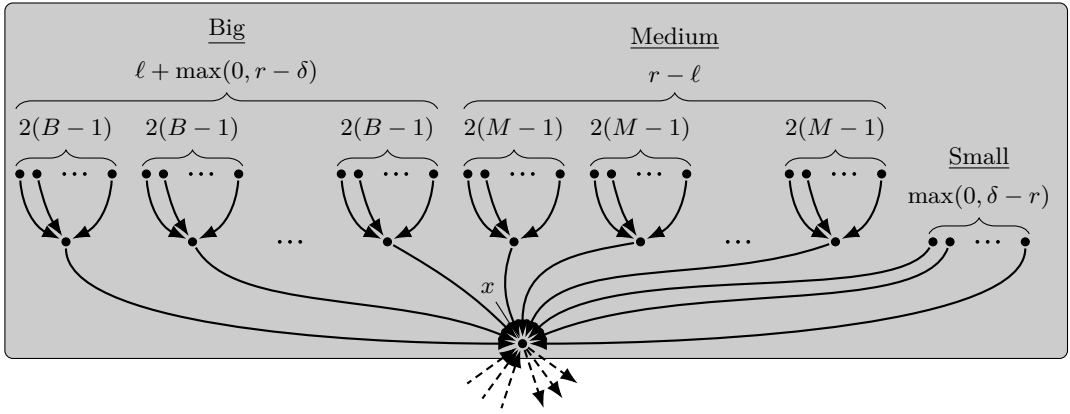

To finish defining the \hsPDD instance $\Instance' = (\Food,d,B,D)$, we need
to specify the types of the quota widgets in \Food, as well as the diversity
function $d$ and the diversity threshold~$D$.  The budget $B$ is the value
defined earlier. We define $d(v_0) = D_0$ and $d(v_1) = d(v_2) = D_1$, for all
$v \in G$. Every other vertex $x \in \Food$ has diversity $d(x) = 1$. The
diversity threshold~$D$ is~$2nD_1 + kD_0$.  Each quota widget $Q_{u,v}$ is of
type $Q(1,1,2,C)$. Each quota widget $Q_v^1$ is of type
$Q(\deg_G(v),0,c(v),C+1)$.  Each quota widget $Q_v^2$ is of type
$Q(\deg_G(v)+1,1,1,1)$.

The construction of $\Instance'$ from $\Instance$ takes polynomial time, as
every vertex and edge of $G$ is replaced by a subgraph of polynomial size, with
appropriate edges between these subgraphs as determined only by the incidence
of vertices and edges in $G$.  Thus, to finish the proof of \cref{thm:tw-hard},
we need to prove that $\tw(\Food) \in O(\tw(G))$ and that $\Instance'$ is a
yes-instance of \hsPDD if and only if \Instance is a yes-instance of \CDS.

\begin{lemma}
	\label{lem:f-tw}
	$\tw(\Food) \le 3\tw(G) + 2$.
\end{lemma}

\begin{proof}
	We consider the underlying undirected graph $U(\Food)$ of \Food, as arc
	directions are unimportant for treewidth.
	
	Given a function $\phi : V(G) \to \NN_{\ge 1}$, a \emph{$\phi$-expansion}
	of $G$ is a graph $G'$ in which each vertex $v$ of $G$ is replaced by a
	clique $C_v$ of size $\phi(v)$ and, for every edge $\{u,v\}$ in $G$, there
	is an edge between every vertex in $C_u$ and every vertex in $C_v$.  Let $a =
	\max_{v \in G} \phi(v)$, then $\tw(G') \le a(\tw(G) + 1) - 1$: Given a tree
	decomposition $(T,\B)$ of width $\tw(G)$ for $G$, the tree decomposition
	$(T,\B')$ in which every bag $B_x \in \B$ is replaced by the bag $B'_x =
	\bigcup_{v \in B_x} C_v$ is a tree decomposition of $G'$ of the desired width.
	
	We can obtain a supergraph $G'$ of $U(\Food)$ from $G$ by (1) subdividing
	each edge $\{u,v\}$ in~$G$ using a vertex $e_{\{u,v\}}$; (2) applying a
	$\phi$-expansion that replaces each vertex $e_{\{u,v\}}$ with two vertices
	$e_{u,v}$ and $e_{v,u}$, and each vertex $v \in V(G)$ with three vertices
	$v_0$, $v_1$, and $v_2$; and (3) attaching to each vertex in the resulting
	graph the respective selector or quota widget of which it is the root.
	Removing edges does not increase a graph's treewidth, so $\tw(\Food) =
	\tw(U(\Food)) \le \tw(G')$.  Similarly, subdividing edges does not increase
	a graph's treewidth~\cite{cygan}, so the graph obtained after Step (1) has
	treewidth $\tw(G)$.  Thus, by the discussion in the previous paragraph,
	the graph obtained after Step (2) has treewidth at most $3\tw(G) + 2$.
	Attaching trees in Step (3) does not increase the treewidth of a graph
	of treewidth at least $1$.  Thus, $\tw(G') \le 3\tw(G) + 2$.
\end{proof}

\begin{lemma}
	\label{lem:valid-reduction}
	\Instance is a \yes-instance of \CDS if and only if $\Instance'$ is a
	\yes-instance of \hsPDD.
\end{lemma}

\begin{proof}
  First, assume that $\Instance$ is a yes-instance of \CDS, and let
  $(S,f)$ be a solution of \CDS on~$\Instance$.  We can assume that $c(v) \le
  \deg_G(v)$, for all~$v \in G$, because $|f^{-1}(v)| \le \deg_G(v)$, for all~$v \in G$.
  We can also assume that $|S| = k$; otherwise, $(S',f|_{V \setminus S'})$,
  for any superset $S' \supsetneq S$ of size $|S'| = k$, is also a feasible
  solution of \CDS on $\Instance$.
  To define a solution $P$ of \hsPDD on~$\Instance'$, we first define a
  viable superset~$P' \supseteq P$ containing the following vertices:
  \begin{itemize}
    \item For each $v \in S$, all vertices in $S_v$.
    \item For each $v \notin S$, all vertices in $Q_{f(v),v}$.
    \item For each $v \in V$, all vertices in $Q_v^1$ and $Q_v^2$.
  \end{itemize}
  To verify that this set is viable, we only need to verify that every root
  of a quota widget in $P'$ has the requisite number of external prey in $P'$,
  as all vertices in the quota widget belong to $P'$.  For every vertex $v_1
  \in P'$, this is trivially true because $Q_v^1$ is of type
  $Q(\deg_G(v),0,c(v),C+1)$, that is, no external prey of $v_1$ are required to
  be included in any viable set.  If $e_{u,v} \in P'$, for some edge
  $\{u,v\}$, then $u = f(v)$, so $u \in S$, that is, $u_0 \in P'$.  Since $u_0$
  is a prey of~$e_{u,v}$ this shows that $e_{u,v}$ has the one external prey in
  $P'$ required by the type~$Q(1,1,2,C)$ of~$Q_{u,v}$.  Finally, every vertex
  $v \in V$ is either in $S$, in which case $v_0 \in P'$, or not, in which
  case~$e_{f(v),v} \in P'$.  Since both are prey of $v_2$, this shows that $v_2$ has
  at least one external prey in~$P'$, as required by the type
  $Q(\deg_G(v)+1,1,1,1)$ of $Q_v^2$.

  For all $v \in V(G)$, let $k_v$ be the number of vertices that $v$ is chosen
  to dominate, that is, $k_v = 0$ if $v \notin S$, and $k_v = |f^{-1}(v)|$ if
  $v \in S$. Then $k_v \le c(v)$, because~$(S,f)$ is a solution of \CDS
  on~\Instance, and~$k_v$ is the number of vertices $e_{u,w} \in P'$ that
  satisfy $u = v$.  In particular, $\sum_{v \in V} k_v = n - k$.  By
  \cref{obs:selector-widget,lem:quota-widget}, there is a \hviable set $P$ in
  \Food that includes the same roots of selector and quota widgets as~$P'$ and
  contains in total
  \begin{itemize}
    \item $C$ vertices from each selector widget $S_v$ such that $v_0 \in P'$,
    \item $C+1$ vertices from each quota widget $Q_{u,v}$ such that $e_{u,v} \in P'$,
    \item $1 + \deg_G(v) + C \cdot c(v) - (C + 1) \cdot k_v$ vertices from each quota widget $Q_v^1$, and
    \item $1 + \deg_G(v)$ vertices from each quota widget $Q_v^2$.
  \end{itemize}
  Thus,
	\[
		\begin{aligned}
			|P|
			&\le C \cdot k + (C + 1)(n - k) + \sum_{v \in V} (2 + 2\deg(v) + C \cdot c(v) - (C + 1)k_v))\\
			&= C \cdot k + 4m + 2n + C \cdot \sum_{v \in V}c(v) = B.
		\end{aligned}
  \]
	Since $P$ contains $v_1$ and $v_2$, for all $v \in V$, and $v_0$, for all $v
	\in S$, the diversity of $P$ is at least~$2n \cdot D_1 + k\cdot D_0 = D$.  Thus,
	$\Instance'$ is a \yes-instance of \hsPDD.
	
  Now, assume that $\Instance'$ is a \yes-instance of \hsPDD, with
  solution~$P$.  A~simple calculation shows that~\Food has at most
  $D_0$ vertices.  Further, \Food has $n$ vertices of diversity $D_0$ and $2n$
  vertices of diversity $D_1$. Thus, if $P$ contained $q < 2n$ of the vertices
  in $X = \{ v_1, v_2 \mid v \in V \}$, then $d(P) < q\cdot D_1 + (n + 1)\cdot
  D_0 \le 2D_1n \le D$.  Consequently, $P$ must contain all vertices in~$X$.
  Analogously, $P$ contains $n_Y \ge k$ of the vertices in $Y = \{ v_0 \mid v
  \in V \}$.
	
	For each $v \in V$, let $k_v$ be the number of vertices $e_{u,w} \in P$
  that satisfy $u = v$.  Then, by \cref{obs:selector-widget,lem:quota-widget},
  $P$ contains at least
  \begin{itemize}
    \item $C$ vertices from each selector widget $S_v$ such that $v_0 \in P$,
    \item $C+1$ vertices from each quota widget $Q_{u,v}$ such that $e_{u,v} \in P$,
    \item $1 + \deg_G(v) + C \cdot \max(c(v),k_v) - (C + 1) \cdot k_v$ vertices from each quota widget $Q_v^1$, and
    \item $1 + \deg_G(v)$ vertices from each quota widget $Q_v^2$.
  \end{itemize}
  Thus,
  \[
    \begin{aligned}
      |P|
      &\ge Cn_Y + \sum_{v \in V}(C + 1)k_v + \sum_{v \in V}(2 + 2\deg_G(v) + C \cdot \max(c(v),k_v) - (C+1) \cdot k_v)\\
      &= C\Bigl(n_Y + \sum_{v \in V}\max(c(v),k_v)\Bigr) + 2n + 4m.
    \end{aligned}
  \]
  Since $|P| \le B = C \cdot k + 4m + 2n + C \cdot \sum_{v \in V} c(v)$, this implies
  that $n_Y \le k$.  We already observed that $n_Y \ge k$, so $n_Y = k$ and,
  therefore, $\sum_{v \in V} \max(c(v),k_v) \le \sum_{v \in V} c(v)$, which
  implies that $k_v \le c(v)$, for all $v \in V$.

  Now, let $S = \{v \in V \mid v_0 \in P\}$.  Then $|S| = n_Y = k$. For each
  vertex $v \notin S$, $v_0 \notin P$ but $v_2 \in P$. By
  \cref{lem:quota-widget}, $P$ contains at least one prey of~$v_2$ not
  in $Q_v^2$. Therefore, since~$v_0 \notin P$, there is at least one
  edge $\{u,v\} \in E(G)$ such that $e_{u,v} \in P$.  By
  \cref{lem:quota-widget} again, this implies that $u_0 \in P$ and, thus, that
  $u \in S$. Thus, we can set $f(v) = u$.  (If there are multiple
  vertices $u$ that meet this condition, we can choose $f(v)$ to be any one of
  them.)
  Since there are $k_u \le c(u)$ vertices $e_{u,v} \in P$, for all $u \in S$,
  this construction of $f$ ensures that $|f^{-1}(u)| \le k_u \le c(u)$, for all
  $u \in S$.  Thus, $(S,f)$ is a feasible solution of \CDS on $\Instance$.
\end{proof}

We have shown that we can, in polynomial time, construct an instance $\Instance'
= (\Food,d,B,D)$ of \hsPDD from an instance $\Instance = (G,k,c)$ of \CDS such that
$\Instance'$ is a \yes-instance if and only if $\Instance$ is and $\tw(\Food)
\in \Oh(\tw(G))$.  Since \CDS is $W[\ell]$-hard, for all $\ell \ge 1$, this shows
that the same is true for \hsPDD, and thereby proves \cref{thm:tw-hard}.

\section{\boldmath \wPDD Parameterized by Node Scanwidth}

In this section, we prove \cref{thm:nsw-fpt}. Let $\Instance = (\Food,
d, B, D)$ be an instance of \wPDD, and let $T$ be a tree extension of $\Food$
of node scanwidth $\nsw(\Food)$.  Let $Z_v$ be the set of all descendants of $v$
in $T$, including~$v$ itself.  We associate with every vertex $v \in T$ a set
$A_v$ of all proper ancestors $u$ of $v$ in $T$ that, in $\Food$, have an arc to
a vertex in $Z_v$.\footnote{$T$ together with the collection
of bags $\B = \{A_v \cup \{v\} \mid v \in T\}$ forms a tree decomposition
$\T = (T,\B)$ of~$\Food$ of width $\nsw(\Food)$~\cite{bruchholdThesis}. We do not use this here.}
These sets satisfy $|A_v| \le \nsw(\Food)$, for all $v \in T$.
For any subset $A' \subseteq A_v$ and any $\ell \in [B]_0$, a set $P \subseteq Z_v$ is
\emph{$(v,A',\ell)$-compatible} if
\begin{enumerate}
	\item $|P| = \ell$, and
	\item Every vertex $w \in P$ is a source in $\Food$ or satisfies
		$\gamma_w(P \cup A') \ge 1$.
\end{enumerate}
In other words, $P \cup A'$ is viable, except that vertices in $A'$ may not have
sufficient prey in $A'$.

We use dynamic programming over $T$ to solve \wPDD on \Instance.  The algorithm
constructs a table $S$ indexed by triples $(v,A',\ell)$ with $v \in T$,
$A' \subseteq A_v$, and $\ell \in [B]_0$.  The entry $S[v,A',\ell]$ stores the
maximum diversity of all $(v,A',\ell)$-compatible sets $P \subseteq Z_v$.
If there is no such set, then $S[v,A',\ell] = -\infty$.

For the root $\rho$ of $T$, we have $A_\rho = \emptyset$ and $Z_\rho = V(\Food)$.
Thus, we immediately obtain the following lemma:

\begin{lemma}
	\label{lem:root-solution}
  The maximum diversity of any viable set $P$ of size at most $B$ in $\Food$
  is~$S[\rho,\emptyset,B]$.  In particular, $\Instance$ is a \yes-instance if
  and only if~$S[\rho,\emptyset,B] \ge D$.
\end{lemma}

\begin{proof}
  If $P^*$ is a set of maximum diversity among all viable sets
  of size at most $B$, then~$P^*$ is $(\rho,\emptyset,|P^*|)$-compatible.  If
  $|P^*| < B$, we can add vertices all of whose prey are in~$P^*$ to~$P^*$
  until $|P^*| = B$.  This does not decrease the diversity of $P^*$ and keeps
  $P^*$ viable. Thus, $d(P^*) \le S[\rho,\emptyset,B]$.

  Conversely, there exists a $(\rho,\emptyset,B)$-compatible set $P$ with $d(P)
  = S[\rho,\emptyset,B]$. Since $A_{\rho} = \emptyset$, $P$ is viable. Since
  $|P| = B$, this shows that $d(P^*) \ge d(P) = S[\rho,\emptyset,B]$.
\end{proof}

It remains to describe how to fill in the table $S$. We do so bottom-up in $T$,
from the leaves to the root.

\subsection{Leaf}

Let~$v$ be a leaf of~$T$.  Then, $Z_v = \{v\}$.  Thus, the only possible entries
$S[v,A',\ell] \ne -\infty$ satisfy $\ell \in \{0,1\}$.  For $\ell = 0$, we have
$S[v,A',0] = 0$, for all $A' \subseteq A_v$, as the empty set is~$(v,A',0)$-compatible,
for all vertices~$v$ and all $A' \subseteq A_v$.  For
$\ell = 1$, we have $S[v,A',1] = d(v)$ if $v$ is a source in~\Food or if $\gamma_v(A') \ge 1$.  Otherwise, $S[v,A',1] = -\infty$.

\subsection{Internal Node}

Let~$v$ be an internal node of~$T$ with children $w_1, \ldots, w_t$ in~$T$.
For every pair $(A',\ell)$ with $A' \subseteq A_v$ and $\ell \in [B]_0$, a set
$P \subseteq Z_v$ is $(v,A',\ell)$-compatible if and only if it can be
partitioned into disjoint sets $P_0, P_1, \ldots, P_t$ such that
\begin{itemize}
  \item $P_0 \subseteq \{v\}$,
  \item $P_i \subseteq Z_{w_i}$, for all $i \in [t]$,
  \item If $P_0 = \{v\}$, then $v$ is a source of $\Food$ or $\gamma_v(A') \ge 1$, and
  \item $P_i$ is $(v,A_{w_i} \cap (A' \cup P_0),|P_i|)$-compatible, for all $i \in [t]$.
\end{itemize}
This holds because, by the properties of a tree extension of $\Food$ and the
definition of the sets~$A_u$, for all $u \in T$, we have $N^-(v) \subseteq A_v$
and, for every $i \in [t]$, $A_{w_i} \subseteq A_v \cup \{v\}$.

Let $\P_0(A')$ be the set of possible choices for $P_0$, that is, $\P_0(A') =
\{\emptyset, \{v\}\}$ if $v$ is a source in $\Food$ or~$\gamma(N^-(v) \cap A')
\ge 1$, and $\P_0(A') = \{\emptyset\}$ otherwise.
Then the characterization of~$(v,A',\ell)$-compatible sets just given implies that
\begin{multline}
  S[v,A',\ell] =
  \max
  \Bigl\{
    d(P_0) + \sum_{i \in [t]} S[w_i,A_{w_i} \cap (A' \cup P_0),\ell_i]
    \,\Bigr|\\
    P_0 \in \P_0(A');
      \ell_1, \ldots, \ell_t \in [\ell]_0;
      |P_0| + \sum_{i \in [t]} \ell_i = \ell
  \Bigr\},
  \label{eq:recurrence}
\end{multline}

To evaluate \cref{eq:recurrence} efficiently for all $\ell \in [B]_0$, we use
an auxiliary table $S_{v,A'}$ indexed by triples $(j,\ell',P_0)$ with $j \in
[t]$, $\ell' \in [B]_0$, and $P_0 \in \P_0(A')$, which satisfies
\begin{align}
  S_{v,A'}[j,\ell',P_0]
  = \max
  \Bigl\{
    \sum_{i \in [j]} S[w_i,A_{w_i} \cap (A' \cup P_0),\ell_i]
  \,\Bigl|\,
    \ell_1, \ldots, \ell_j \in [\ell']_0, \sum_{i \in [j]} \ell_i = \ell'
  \Bigr\}.
\end{align}
Then
\begin{align}
  S[v,A',\ell] &= \max_{\mathclap{P_0 \in \P_0(A')}} (d(P_0) + S_{v,A'}[t,\ell - |P_0|,P_0]),\label{eq:aux-to-S}
\intertext{and we can compute these auxiliary table entries using the recurrence}
  S_{v,A'}[1,\ell',P_0] &= S[w_1,A_{w_1} \cap (A' \cup P_0),\ell'],\label{eq:aux-init}\\
  S_{v,A'}[j,\ell',P_0] &= \max_{\mathclap{\ell_j \in [\ell']_0}} (S_{v,A'}[j-1,\ell' - \ell_j,P_0] + S[w_j,A_{w_j} \cap (A' \cup P_0),\ell_j,P_0]) \ \rlap{$\forall j > 1$.} \label{eq:aux-induction}
\end{align}

Given the auxiliary tables, each entry in $S$ can be computed in constant time
using \cref{eq:aux-to-S}. Since $|A_v| \le \nsw(\Food)$ and $\ell \le B \le
|V(\Food)|$, there are at most $n^2 \cdot 2^{\nsw(\Food)}$ such entries to
compute, where $n = |V(\Food)|$.  Thus, filling in $S$ takes $O\left(n^2 \cdot
2^{\nsw(\Food)}\right)$ time.

In total, there are $O\left(n^2 \cdot 2^{\nsw(\Food)}\right)$ auxiliary table
entries: There are at most $2(B + 1) \le 2(n + 1)$ auxiliary table entries for
each pair $(v,j)$ and each subset $A' \subseteq A_v$.  Each pair~$(v,j)$
corresponds to a child of $v$, so there are $n - 1$ such pairs for all nodes in
$T$.  There are $2^{|A_v|} \le 2^{\nsw(\Food)}$ choices for~$A'$, for each
pair~$(v,j)$. Since each auxiliary table entry can be computed in $O(n)$ time
using \cref{eq:aux-init,eq:aux-induction}, filling in all auxiliary tables
takes~$O\left(n^3 \cdot 2^{\nsw(\Food)}\right)$ time. This completes the proof
of \cref{thm:nsw-fpt}.

\section{Discussion}

This paper establishes the first explicit separation between treewidth and
scanwidth from the perspective of fixed-parameter tractability. While scanwidth
has been recognized as a good parameterization for obtaining efficient
algorithms in phylogenetics and related biological fields, its algorithmic
power relative to treewidth had remained an open question. Our results prove
that scanwidth is not merely a larger variant of treewidth, but enables
efficient algorithms even in settings where treewidth-based approaches fail.

Several open questions remain. First, it would be interesting to identify
further problems providing an analogous separation between these two width
measures. Second, while our algorithm assumes that a tree extension of bounded
node scanwidth is given, the computational complexity of finding such a tree
extension efficiently remains open for node scanwidth, and computing a tree
extension of optimal (edge) scanwidth is currently only known to be in \XP,
when parameterized by the scanwidth. Exploring relationships between scanwidth
and other directed width measures may yield additional insights into the
parameterized complexity of problems on directed acyclic graphs.

%

\newpage
\pagestyle{empty}
\bibliography{ref}

@book{cygan,
  title = {Parameterized {A}lgorithms},
  author = {Cygan,Marek and Fomin, Fedor V. and Kowalik, Lukasz and Lokshtanov,
            Daniel and Marx, D{\'{a}}niel and Pilipczuk, Marcin and Pilipczuk,
            Michal and Saurabh, Saket},
  publisher = {Springer},
  year = {2015},
  doi = {10.1007/978-3-319-21275-3},
}

@book{downeybook,
  title = {Fundamentals of {P}arameterized {C}omplexity},
  author = {Downey, Rodney G. and Fellows, Michael R.},
  series = {Texts in Computer Science},
  publisher = {Springer},
  year = {2013},
  doi = {10.1007/978-1-4471-5559-1},
}

@book{diestel2025graph,
  title = {{Graph Theory}},
  author = {Diestel, Reinhard},
  volume = {173},
  year = {2025},
  publisher = {Springer Nature},
  doi = {10.1007/978-3-662-70107-2},
}

@article{robertson1986graph,
  title = {{Graph Minors. II. Algorithmic Aspects of Tree-Width}},
  author = {Robertson, Neil and Seymour, Paul D.},
  journal = {Journal of Algorithms},
  volume = {7},
  number = {3},
  pages = {309--322},
  year = {1986},
  publisher = {Elsevier},
  doi = {10.1016/0196-6774(86)90023-4},
}

@inproceedings{bodlaender2007treewidth,
  title = {{Treewidth: Structure and Algorithms}},
  author = {Bodlaender, Hans L.},
  booktitle = {Proceedings of the 14th International Conference on Structural
               Information and Communication Complexity (SIROCCO 2007)},
  pages = {11--25},
  year = {2007},
  doi = {10.1007/978-3-540-72951-8_3},
}

@inproceedings{bodlaender1988dynamic,
  title = {{Dynamic Programming on Graphs with Bounded Treewidth}},
  author = {Bodlaender, Hans L.},
  booktitle = {Proceedings of the 15th International Colloquium on Automata,
               Languages and Programming (ICALP 1988)},
  pages = {105--118},
  year = {1988},
  doi = {10.1007/3-540-19488-6_110},
}

@article{fellows2011complexity,
  title = {{On the complexity of some colorful problems parameterized by
           treewidth}},
  author = {Fellows, Michael R. and Fomin, Fedor V. and Lokshtanov, Daniel and
            Rosamond, Frances and Saurabh, Saket and Szeider, Stefan and
            Thomassen, Carsten},
  journal = {Information and Computation},
  volume = {209},
  number = {2},
  pages = {143--153},
  year = {2011},
  publisher = {Elsevier},
  doi = {10.1016/j.ic.2010.11.026},
}

@article{FAITH1992,
  title = {{Conservation evaluation and phylogenetic diversity}},
  author = {Faith, Daniel P.},
  journal = {Biological Conservation},
  volume = {61},
  number = {1},
  pages = {1-10},
  year = {1992},
  doi = {10.1016/0006-3207(92)91201-3}
}

@article{steel,
  title = {{Phylogenetic Diversity and the Greedy Algorithm}},
  author = {Steel, Mike},
  journal = {Systematic Biology},
  volume = {54},
  number = {4},
  pages = {527--529},
  year = {2005},
  publisher = {Society of Systematic Zoology},
  doi = {10.1080/10635150590947023},
}

@article{Pardi2005,
  title = {{Species Choice for Comparative Genomics: Being Greedy Works}},
  author = {Pardi, Fabio and Goldman, Nick},
  journal = {PLoS Genetics},
  year = {2005},
  volume = {1},
  doi = {10.1371/journal.pgen.0010071},
}

@article{faller,
  title = {{Optimizing Phylogenetic Diversity with Ecological Constraints}},
  author = {Faller, Be{\'a}ta and Semple, Charles and Welsh, Dominic},
  journal = {Annals of Combinatorics},
  volume = {15},
  number = {2},
  pages = {255--266},
  year = {2011},
  publisher = {Springer},
  doi = {10.1007/s00026-011-0093-6},
}

@article{moulton,
  title = {{Optimizing phylogenetic diversity under constraints}},
  author = {Moulton, Vincent and Semple, Charles and Steel, Mike},
  journal = {Journal of Theoretical Biology},
  volume = {246},
  number = {1},
  pages = {186--194},
  year = {2007},
  publisher = {Elsevier},
  doi = {10.1016/j.jtbi.2006.12.021},
}

@inproceedings{berry,
  title = {Scanning phylogenetic networks is {NP}-hard},
  author = {Berry, Vincent and Scornavacca, Celine and Weller, Mathias},
  booktitle = {Proceedings of the 46th International Conference on Current
               Trends in Theory and Practice of Computer Science (SOFSEM 2020)},
  pages = {519--530},
  year = {2020},
  organization = {Springer},
  doi = {10.1007/978-3-030-38919-2_42},
}

@mastersthesis{bruchholdThesis,
  author = {Bruchhold, Sebastian},
  title = {{On (Node-)Scanwidth and Its Algorithmic Applications}},
  school = {Technische Universität Berlin},
  year = {2024},
  type = {Master's Thesis},
  doi = {10.14279/depositonce-21448},
}

@mastersthesis{holtgrefeThesis,
  author = {Holtgrefe, Niels},
  title = {{Computing the Scanwidth of Directed Acyclic Graphs}},
  school = {Technische Universiteit Delft},
  year = {2023},
  type = {Master's Thesis},
  url = {https://resolver.tudelft.nl/uuid:9c82fd2a-5841-4aac-8e40-d4d22542cdf5},
}

@article{holtgrefe2024exact,
  title = {{Exact and Heuristic Computation of the Scanwidth of Directed Acyclic
           Graphs}},
  author = {Holtgrefe, Niels and van Iersel, Leo and Jones, Mark},
  journal = {arXiv preprint arXiv:2403.12734},
  year = {2024},
  eprint = {2403.12734},
  archivePrefix = {arXiv},
}

@article{scornavacca2022treewidth,
  title = {{Treewidth-based algorithms for the small parsimony problem on
           networks}},
  author = {Scornavacca, Celine and Weller, Mathias},
  journal = {Algorithms for Molecular Biology},
  volume = {17},
  number = {1},
  pages = {15},
  year = {2022},
  publisher = {Springer},
  doi = {10.1186/s13015-022-00216-w},
}

@inproceedings{bruchhold2025exploiting,
  title = {{Exploiting Low Scanwidth to Resolve Soft Polytomies}},
  author = {Bruchhold, Sebastian and Weller, Mathias},
  booktitle = {Proceedings of the 51st International Conference on Current
               Trends in Theory and Practice of Computer Science (SOFSEM 2026)},
  DELETEpages = {ToDo},
  year = {2026},
  organization = {Springer},
  eprint = {2511.20771},
  archivePrefix = {arXiv},
}

@inproceedings{MAPPD,
  title = {{How Can We Maximize Phylogenetic Diversity? Parameterized Approaches
           for Networks}},
  author = {Jones, Mark and Schestag, Jannik},
  booktitle = {Proceedings of the 18th International Symposium on Parameterized
               and Exact Computation (IPEC 2023)},
  pages = {30:1--30:12},
  year = {2023},
  organization = {Schloss-Dagstuhl-Leibniz Zentrum f{\"u}r Informatik},
  doi = {10.4230/LIPIcs.IPEC.2023.30},
}

@inproceedings{MAPPDD,
  title = {{Parameterized Algorithms for Diversity of Networks with Ecological
           Dependencies}},
  author = {Jones, Mark and Schestag, Jannik},
  booktitle = {Proceedings of the 20th International Symposium on Parameterized
               and Exact Computation (IPEC 2025)},
  pages = {11:1--11:21},
  year = {2025},
  organization = {Schloss-Dagstuhl-Leibniz Zentrum f{\"u}r Informatik},
  doi = {10.4230/LIPIcs.IPEC.2025.11},
}

@article{PaNDA,
  title = {{PaNDA: Efficient Optimization of Phylogenetic Diversity in Networks}},
  author = {Holtgrefe, Niels and van Iersel, Leo and Meuwese, Ruben and Murakami
            , Yukihiro and Schestag, Jannik},
  year = {2025},
  publisher = {Cold Spring Harbor Laboratory},
  journal = {bioRxiv: 10.1101/2025.11.14.688467},
  doi = {10.1101/2025.11.14.688467},
}

@inproceedings{PDD,
  title = {{Maximizing Phylogenetic Diversity under Ecological Constraints: A
           Parameterized Complexity Study}},
  author = {Komusiewicz, Christian and Schestag, Jannik},
  booktitle = {Proceedings of the 44th IARCS Annual Conference on Foundations of
               Software Technology and Theoretical Computer Science (FSTTCS~2024)
               },
  year = {2024},
  pages = {28:1--28:18},
  organization = {Schloss-Dagstuhl-Leibniz Zentrum f{\"u}r Informatik},
  doi = {10.4230/LIPIcs.FSTTCS.2024.28},
}

@inproceedings{WeightedFW,
  title = {{Weighted Food Webs Make Computing Phylogenetic Diversity So Much
           Harder}},
  author = {Jannik Schestag},
  booktitle = {Proceedings of the 51st International Conference on Current
               Trends in Theory and Practice of Computer Science (SOFSEM 2026)},
  DELETEpages = {ToDo},
  year = {2026},
  organization = {Springer},
  eprint = {2510.05911},
  archivePrefix = {arXiv},
}

@inproceedings{1PDDkernel,
  title = {{Limits of Kernelization and Parametrization for Phylogenetic
           Diversity with Dependencies}},
  author = {Niels Holtgrefe and Jannik Schestag and Norbert Zeh},
  booktitle = {Proceedings of the 17th Latin American Theoretical Informatics
               Symposium (LATIN 2026)},
  DELETEpages = {ToDo},
  year = {2026},
  organization = {Springer},
  TODODOI = {},
}

@article{bodlaender2025xnlp,
  title = {{XNLP-Completeness for Parameterized Problems on Graphs with a Linear
           Structure}},
  author = {Bodlaender, Hans L. and Groenland, Carla and Jacob, Hugo and Jaffke,
            Lars and Lima, Paloma T.},
  journal = {Algorithmica},
  volume = {87},
  number = {4},
  pages = {465--506},
  year = {2025},
  publisher = {Springer},
  doi = {10.1007/s00453-024-01274-9},
}

@article{vanIerselTreeContainment,
  title = {Embedding phylogenetic trees in networks of low treewidth},
  author = {van Iersel, Leo and Jones, Mark and Weller, Mathias},
  sortname = {Iersel, Leo van and Jones, Mark and Weller, Mathias},
  journal = {Discrete Mathematics \& Theoretical Computer Science},
  year = {2023},
  volume = 45,
  pages = 2,
  doi = {10.46298/dmtcs.10116},
}

@unpublished{JonesTreeContainment,
  title = {Tree Contaiment Parameterized By ScanWidth},
  author = {Jones, Mark and Weller, Mathias},
  year = {2026},
  note = {Manuscript in preparation},
}
\bibliographystyle{plainurl}   

\end{document}